\documentclass[a4paper,10pt]{amsart}
\usepackage{graphicx}
\usepackage{amssymb,amsmath,amstext,amscd}
\theoremstyle{plain}
\newtheorem{thm}{Theorem}[section]

\newtheorem{corollary}{Corollary}[section]
\newtheorem*{main theorem}{Theorem}

\theoremstyle{definition}

\begin{document}

\title{A characterization of causal automorphisms by wave equations}
\author{Do-Hyung Kim}
\address{Department of Applied Mathematics, College of Advanced Science, Dankook University,
San 29, Anseo-dong, Dongnam-gu, Cheonan-si, Chungnam, 330-714,
Republic of Korea} \email{mathph@dankook.ac.kr}


\begin{abstract}
A characterization of causal automorphism on Minkowski spacetime
is given by use of wave equation. The result shows that causal
analysis of spacetime may be replaced by studies of wave equation
on manifolds.
\end{abstract}

\maketitle

\section{Introduction} \label{section:1}

By causal automorphism, we mean a bijection between spacetimes
that preserves causal relations. In 1964, Zeeman has shown that
causal automorphisms on $\mathbb{R}^{n+1}_1$ with $n \geq 2$ are
generated by inhomogeneous Lorentz group together with the
dilatation.(Ref. \cite{Zeeman}). As Zeeman commented, his theorem
does not hold in two-dimensional case. Recently, the solution to
two-dimensional case was given.(Ref. \cite{CQG3}, \cite{CQG4},
\cite{Minguzzi} and \cite{Low}) As Low has commented in
\cite{Low}, each component of causal automorphisms on
$\mathbb{R}^{n}_1$ satisfies the wave equation and so it is
natural to ask the relationship between wave equation and causal
relation.

In this paper, we characterize causal automorphisms on
$\mathbb{R}^{n+1}_1$ by wave equations. This gives a partial
answer to the question raised by Low.

\section{The case :  $\mathbb{R}^{n+1}_1$ with $n \geq 2$} \label{section:2}

We use the signature $(+, \cdots, +, -)$ for the metric of
$\mathbb{R}^{n+1}_1$.

In general, the wave equation $\sum\limits_{i=1}^n
\frac{\partial^2{\varphi}}{\partial x_i^2} - \frac{1}{v^2}
\frac{\partial^2{\varphi}}{\partial t^2} = 0$ represents a wave
with propagating speed $v$. From postulates in the theory of
relativity, we assume that $v=c=1$ and by setting $t=x_{n+1}$, the
wave equation becomes $\sum\limits_i^n \frac{\partial^2
\varphi}{\partial x_i^2} - \frac{\partial^2 \varphi}{\partial
x_{n+1}^2} = 0$. If we define $\epsilon_i$ by $\epsilon_1 = \cdots
= \epsilon_n = 1$ and $\epsilon_{n+1}=-1$, then the wave equation
can be written by $\sum\limits_{j,k=1}^{n+1}
\frac{\partial^2{\varphi}}{\partial y_j \partial y_k} \delta_{jk}
\epsilon_k =0$ or $\sum\limits_{i=1}^n \epsilon_i \frac{\partial^2
\varphi}{\partial x_i^2}=0$.

\begin{thm} \label{lemma}
Let $(x_1, \cdots, x_{n+1})$ be the standard coordinate system on
$\mathbb{R}^{n+1}_1$  and $(y_1, \cdots, y_{n+1})$ be another
coordinate system on $\mathbb{R}^{n+1}_1$ with $n \geq 2$. Assume
that, for any smooth function $\varphi$, $\sum\limits_{i=1}^n
\frac{\partial^2 \varphi}{\partial x_i^2} - \frac{\partial^2
\varphi}{\partial x_{n+1}^2} = 0$ if and only if
$\sum\limits_{i=1}^n \frac{\partial^2 \varphi}{\partial y_i^2} -
\frac{\partial^2 \varphi}{\partial y_{n+1}^2} = 0$. Then, all rows
of Jacobian matrix $\big( \frac{\partial y_i}{\partial x_j} \big)$
is mutually orthogonal and have the same length. Furthermore, the
first $n$ rows are spacelike vectors and the last row is a
timelike vector.
\end{thm}
\begin{proof}
By chain rule, we have \\
$\frac{\partial \varphi}{\partial x_i} = \sum\limits_{j=1}^{n+1}
\frac{\partial \varphi}{\partial y_j}\frac{\partial y_j}{\partial
x_i}$ and,\\
$\frac{\partial^2 \varphi}{\partial x_i^2} = \sum\limits_{j,k}
\frac{\partial^2 \varphi}{\partial y_k \partial y_j}\frac{\partial
y_k}{\partial x_i}\frac{\partial y_j}{\partial x_i} +
\sum\limits_j \frac{\partial \varphi}{\partial
y_j}\frac{\partial^2 y_j}{\partial x_i^2}$.

Therefore, we have

\begin{eqnarray*}
\sum\limits_{i=1}^n \frac{\partial^2 \varphi}{\partial x_i^2} -
\frac{\partial^2 \varphi}{\partial x_{n+1}^2} & = &
\sum\limits_{i=1}^n \bigg(\sum\limits_{j,k=1}^{n+1}
\frac{\partial^2 \varphi}{\partial y_k
\partial y_j}\frac{\partial y_k}{\partial x_i}\frac{\partial
y_j}{\partial x_i} + \sum\limits_{j=1}^{n+1} \frac{\partial
\varphi}{\partial y_j}\frac{\partial^2 y_j}{\partial x_i^2}\bigg)\\
& & - \bigg(\sum\limits_{j,k=1}^{n+1} \frac{\partial^2
\varphi}{\partial y_k
\partial y_j}\frac{\partial y_k}{\partial x_{n+1}}\frac{\partial
y_j}{\partial x_{n+1}} + \sum\limits_{j=1}^{n+1} \frac{\partial
\varphi}{\partial y_j}\frac{\partial^2 y_j}{\partial x_{n+1}^2}\bigg)\\
& = & \sum\limits_{j,k=1}^{n+1} \frac{\partial^2 \varphi}{\partial
y_k
\partial y_j}\bigg( \sum\limits_{i=1}^n \frac{\partial y_k}{\partial
x_i}\frac{\partial y_j}{\partial x_i} - \frac{\partial
y_k}{\partial x_{n+1}}\frac{\partial y_j}{\partial x_{n+1}}\bigg) \\
& & + \sum\limits_{j=1}^{n+1} \frac{\partial \varphi}{\partial
y_j}\bigg(\sum\limits_{i=1}^n \frac{\partial^2 y_j}{\partial
x_i^2}-\frac{\partial^2 y_j}{\partial x_{n+1}^2}\bigg) \\
&=& 0.
\end{eqnarray*}

From $\sum\limits_{j,k=1}^{n+1} \frac{\partial^2 \varphi}{\partial
y_j \partial y_k} \delta_{jk} \epsilon_k =0$, we have
$$\sum\limits_{j,k=1}^{n+1} \frac{\partial^2 \varphi}{\partial y_k
\partial y_j} \Big( \sum_{i=1}^n \frac{\partial y_k}{\partial
x_i}\frac{\partial y_j}{\partial x_i} - \frac{\partial
y_k}{\partial x_{n+1}} \frac{\partial y_j}{\partial x_{n+1}} -
\delta_{jk} \epsilon_k \Big) + \sum\limits_{j=1}^{n+1}
\frac{\partial \varphi}{\partial y_j} \Big( \sum\limits_{i=1}^n
\frac{\partial^2 y_j}{\partial x_i^2} - \frac{\partial^2
y_j}{\partial x_{n+1}^2} \Big) = 0. $$

Since each $y_i$ satisfies $\sum\limits_{i=1}^{n+1} \epsilon_i
\frac{\partial^2 \varphi}{\partial y_i^2} = 0$, by the given
assumption, it also satisfies the equation
$\sum\limits_{i=1}^{n+1} \epsilon_i \frac{\partial^2
\varphi}{\partial x_i^2} = 0$, and thus the second terms of the
above equation must vanish and thus the above equation becomes
$$\sum\limits_{j,k=1}^{n+1} \frac{\partial^2 \varphi}{\partial y_k
\partial y_j} \Big( \sum_{i=1}^n \frac{\partial y_k}{\partial
x_i}\frac{\partial y_j}{\partial x_i} - \frac{\partial
y_k}{\partial x_{n+1}} \frac{\partial y_j}{\partial x_{n+1}} -
\delta_{jk} \epsilon_k \Big)=0 \cdots \cdots (*).$$

We now show the terms in the parenthesis of $(*)$ vanish for $j
\neq k$, which shows that all rows of $\big( \frac{\partial
y_i}{\partial x_j} \big)$ are mutually orthogonal.

Let $\varphi(y_1, \cdots, y_{n+1}) = \exp(a_1y_1 + \cdots +
a_{n+1}y_{n+1})$ with $a_1^2 + \cdots + a_n^2 = a_{n+1}^2$. Then,
$\frac{\partial^2 \varphi}{\partial y_j \partial y_k} = a_ja_k
\exp(a_1y_1 + \cdots + a_{n+1}y_{n+1})$, and since $\varphi$ is a
solution of $\sum\limits_{i=1}^{n+1} \epsilon_i \frac{\partial^2
\varphi}{\partial y_i^2} = 0$, it must be that
$\sum\limits_{i=1}^{n+1} \epsilon_i \frac{\partial^2
\varphi}{\partial x_i^2} = 0$. Therefore, $\varphi$ must satisfy
the equation $(*)$ for any real numbers $a_i$'s with $a_1^2 +
\cdots + a_n^2 = a_{n+1}^2$. For the sake of simplicity, we denote
the term $ \sum_{i=1}^n \frac{\partial y_k}{\partial
x_i}\frac{\partial y_j}{\partial x_i} - \frac{\partial
y_k}{\partial x_{n+1}} \frac{\partial y_j}{\partial x_{n+1}} -
\delta_{jk} \epsilon_k $ by $( k , j )$.

For $\alpha >0$, by putting $a_i = a_{n+1} = \alpha$ and $-a_i =
a_{n+1} = \alpha$ in $(*)$, we have two equations. By subtracting
two equations, we have $(i , n+1) = 0$, which means that all
$i$-th rows of $\big( \frac{\partial y_i}{\partial x_j} \big)$ are
orthogonal to the $(n+1)$-th row for $1 \leq i \leq n$.

For $i \neq j$, by putting $a_i = \frac{\alpha}{\sqrt{2}}$, $a_j =
\frac{\alpha}{\sqrt{2}}$, $a_{n+1} = \alpha$ in $(*)$ and then by
putting $a_i = \frac{\alpha}{\sqrt{2}}$, $a_j =
-\frac{\alpha}{\sqrt{2}}$, $a_{n+1} = \alpha$ in $(*)$, we have
two equations. By subtracting them, we have $(i,j)=0$, which means
that all $i$-th and $j$-th rows are orthogonal for $1 \leq i \neq
j \leq n$.

By considering the above two results, and if we put $a_i = a_{n+1}
= \alpha$ in $(*)$, we have $(i,i)+(n+1,n+1)=0$ for $1 \leq i \leq
n$, which implies that all rows have the same length.

From $(i,i)+(n+1,n+1)=0$, we can see that if $(i, i)=0$, then
$(n+1, n+1)=0$. In other words, all rows of the Jacobian matrix
are null vectors. In general, mutually orthogonal null vectors are
co-linear, this implies that all rows of the Jacobian are linearly
dependent, which implies that $\big( \frac{\partial y_i}{\partial
x_j} \big)$ is singular. This is a contradiction. Therefore, the
first $n$ rows are spacelike vectors and the $(n+1)$-th row is a
timelike vector.
\end{proof}

In 1964, Zeeman has shown the following theorem. (Ref.
\cite{Zeeman}).

\begin{thm}
Let $F : \mathbb{R}^{n+1}_1 \rightarrow \mathbb{R}^{n+1}_1$ be a
causal automorphism with $n \geq 2$. Then, there exist a real
number $a$, an orthochronous matrix $A$ and $b \in
\mathbb{R}^{n+1}_1$ such that $F(x) = a \cdot A x +b$.
\end{thm}

We now characterize causal automorphisms on $\mathbb{R}^{n+1}_1$.

\begin{thm} \label{main}
Let $F : \mathbb{R}^{n+1}_1 \rightarrow \mathbb{R}^{n+1}_1$ given
by $(y_1, \cdots, y_{n+1})= F(x_1, \cdots, x_{n+1})$ be a
diffeomorphism with $n \geq 2$. Then the necessary and sufficient
condition for $F$ to be a causal automorphism on
$\mathbb{R}^{n+1}_1$ is that $\frac{\partial y_{n+1}}{\partial
x_{n+1}} \geq 0$, and for any smooth function $\varphi$,
$\sum\limits_{i=1}^n \frac{\partial^2 \varphi}{\partial x_i^2} -
\frac{\partial^2 \varphi}{\partial x_{n+1}^2} = 0$ if and only if
$\sum\limits_{i=1}^n \frac{\partial^2 \varphi}{\partial y_i^2} -
\frac{\partial^2 \varphi}{\partial y_{n+1}^2} = 0$.
\end{thm}
\begin{proof}
Assume that $F$ is a causal automorphism. Then, by Zeeman's
theorem, we have $y_j = \alpha \sum\limits_{k=1}^{n+1} a_{jk}x_k +
b_k$ where $(a_{jk})$ is an orthochronous matrix. Then, we have \\
$\frac{\partial \varphi}{\partial x_i} = \sum\limits_j
\frac{\partial \varphi}{\partial y_j} \frac{\partial y_j}{\partial
x_i} = \sum\limits_j \frac{\partial \varphi}{\partial y_j}(\alpha
a_{ji})$, \\
$\frac{\partial^2 \varphi}{\partial x_i^2} = \sum\limits_{j,k}
\alpha^2 a_{ji}a_{ki} \frac{\partial^2 \varphi}{\partial y_k
\partial y_j}$.
Therefore, we have
\begin{eqnarray*}
\sum\limits_i \epsilon_i \frac{\partial^2 \varphi}{\partial x_i^2}
&=& \sum\limits_{i,j,k} \epsilon_i \alpha^2 a_{ji}a_{ki}
\frac{\partial^2 \varphi}{\partial y_k \partial y_j} \\
&=& \sum\limits_{j,k} \alpha^2 \delta_{jk} \epsilon_j
\frac{\partial^2 \varphi}{\partial y_k \partial y_j} \\
&=& \alpha^2 \Big( \sum\limits_j \epsilon_j \frac{\partial^2
\varphi}{\partial y_j^2}\Big).
\end{eqnarray*}
In other words, $\sum\limits_{i=1}^n \frac{\partial^2
\varphi}{\partial x_i^2} - \frac{\partial^2 \varphi}{\partial
x_{n+1}^2} = 0$ if and only if $\sum\limits_{i=1}^n
\frac{\partial^2 \varphi}{\partial y_i^2} - \frac{\partial^2
\varphi}{\partial y_{n+1}^2} = 0$. Since $(a_{ij})$ is
orthochronous, we must have $a_{n+1,n+1} \geq 1$, and thus
$\frac{\partial y_{n+1}}{\partial x_{n+1}} = \alpha a_{n+1,n+1}
>0$.

Conversely, assume that $\sum\limits_{i=1}^n \frac{\partial^2
\varphi}{\partial x_i^2} - \frac{\partial^2 \varphi}{\partial
x_{n+1}^2} = 0$ if and only if $\sum\limits_{i=1}^n
\frac{\partial^2 \varphi}{\partial y_i^2} - \frac{\partial^2
\varphi}{\partial y_{n+1}^2} = 0$ and $\frac{\partial
y_{n+1}}{\partial x_{n+1}} \geq 1$. Then, by the theorem
\ref{lemma}, $\frac{1}{f(x_1, \cdots, x_{n+1})} \big(
\frac{\partial y_i}{\partial x_j} \big)$ is a Lorentz matrix for
the smooth function $f = < \mbox{grad} \,\, y_1, \mbox{grad} \,\,
y_1
>^{\frac{1}{2}}$. If we let $Y_i = \frac{1}{f} y_i$, then
 since $\big( \frac{\partial
x_i}{\partial Y_j} \big) = \big( \frac{\partial x_i}{\partial y_k}
\big) \big( \frac{\partial y_k}{\partial Y_j} \big) = f \big(
\frac{\partial y_i}{\partial x_j} \big)^{-1}$, $\big(
\frac{\partial x_i}{\partial Y_j} \big)$ is a Lorentz matrix.

If we consider $(x_1, \cdots, x_{n+1}) \mapsto (Y_1, \cdots,
Y_{n+1})$ as a map from $\mathbb{R}^{n+1}_1$ into
$\mathbb{R}^{n+1}_1$, then since $\big( \frac{\partial
Y_i}{\partial x_j} \big)$ is a Lorentz matrix, the map is an
isometry from $\mathbb{R}^{n+1}_1$ into a subset of
$\mathbb{R}^{n+1}_1$. Since $\mathbb{R}^{n+1}_1$ can not be
isometric to its proper subset, the map is surjective and thus the
map is an isometry from $\mathbb{R}^{n+1}_1$ onto
$\mathbb{R}^{n+1}_1$. Therefore, by proposition 10 of chapter 9 in
Ref. \cite{Oneill}, we have $(Y_1, \cdots, Y_{n+1})^t = A (x_1,
\cdots, x_{n+1})^t + b$ for some $(n+1)$-by-$(n+1)$ Lorentz matrix
$A$. Therefore, we have $(y_1, \cdots, y_{n+1})^t = f A (x_1,
\cdots, x_{n+1})^t + b$.

We now show that the function $f$ is a constant function. From
$\frac{1}{f} \frac{\partial y_{n+1}}{\partial x_j} = \frac{1}{f}
\frac{\partial f}{\partial x_j} \big( \sum\limits_{k=1}^{n+1}
a_{n+1,k}x_k \big) + a_{n+1,j}$, since the vectors $\big(
\frac{1}{f} \frac{\partial y_{n+1}}{\partial x_j} \big)$ and
$\big( a_{n+1,j} \big)$ are timelike unit vectors, we can see that
grad $f$ is spacelike at all points. In the proof of theorem
\ref{lemma}, we have seen that the condition, $\sum\limits_{i=1}^n
\frac{\partial^2 \varphi}{\partial x_i^2} - \frac{\partial^2
\varphi}{\partial x_{n+1}^2} = 0$ if and only if
$\sum\limits_{i=1}^n \frac{\partial^2 \varphi}{\partial y_i^2} -
\frac{\partial^2 \varphi}{\partial y_{n+1}^2} = 0$, implies that
for each $i$, $y_i$ must satisfy $\sum\limits_{j=1}^{n+1}
\epsilon_j \frac{\partial^2 y_i}{\partial x_j^2} = 0$. If we
substitute $y_i = f \sum\limits_{k=1}^{n+1} a_{ik} x_k + b_i$ into
$\sum\limits_{j=1}^{n+1} \epsilon_j \frac{\partial^2 y_i}{\partial
x_j^2} = 0$, we have, $$ \sum\limits_{k=1}^{n+1} a_{ik} \Big[
\big( \sum\limits_{j=1}^{n+1} \epsilon_j \frac{\partial^2
f}{\partial x_j^2} \big) x_k + 2 \epsilon_k \frac{\partial
f}{\partial x_k} \Big] = 0, \,\, \mbox{for each} \,\, i. \cdots
\cdots (**)$$ Since $(a_{ik})$ is non-singular, we must have
$\big( \sum\limits_{j=1}^{n+1} \epsilon_j \frac{\partial^2
f}{\partial x_j^2} \big) x_k + 2 \epsilon_k \frac{\partial
f}{\partial x_k} = 0$ for all $k$. In other words, $\big(
\sum\limits_{j=1}^{n+1} \epsilon_j \frac{\partial^2 f}{\partial
x_j^2} \big)$\textbf{x}$+ 2 \, \mbox{grad}  f = 0$ where
\textbf{x} $= (x_1, \cdots, x_{n+1})$ is a position vector.

%=======================================================================%
 Since grad $f$ is spacelike, if we take \textbf{x} to be timelike, we
have a contradiction. Therefore, grad $f = 0$ and
$\sum\limits_{j=1}^{n+1} \epsilon_j \frac{\partial^2 f}{\partial
x_j^2}=0$ at each point in future and past time cone with apex at
the origin. By continuity, they are zero on the closure of the
time cone.

%======================================================================%

We let $ \widetilde{f} = \sum\limits_{j=1}^{n+1} \epsilon_j
\frac{\partial^2 f}{\partial x_j^2}$, and take divergence of both
sides of the previous equation, we have \textbf{x} $[
\widetilde{f} ] + (n+3) \widetilde{f} =0$. This is the typical
example of first-order partial differential equation, called
Euler's equation, and it has a unique solution $0$.(See Section 6
of Chapter 1 in Ref. \cite{John}). Therefore, $\widetilde{f}$
vanishes and from the equation $(**)$, we have that $\mbox{grad}
f$ is constantly zero, and thus $f$ is a positive real number
$\alpha$.

Finally, since $\frac{\partial y_{n+1}}{\partial x_{n+1}} = \alpha
a_{n+1,n+1} > 0$, we have $a_{n+1,n+1} > 1$ and thus $A$ is an
orthochronous matrix. In conclusion, by Zeeman's theorem again,
the map $F$ is a causal automorphism.

\end{proof}

It is a well-known fact that every $C^2$ function that satisfies
Laplace equation is actually a $C^\infty$ function. We can state a
similar result by use of above results. In previous theorems, we
have tacitly assumed that $\varphi$ and coordinate transformation
$(x_1, \cdots, x_{n+1}) \mapsto (y_1, \cdots, y_{n+1})$ are
$C^\infty$. However, as can be seen in the proof, it suffices to
assume that $\varphi$ and the coordinate transformation are $C^2$.

 Then, we have the following corollary.

\begin{corollary}
Let $(x_1, \cdots, x_{n+1})$ be the standard coordinate system on
$\mathbb{R}^{n+1}$ and $(y_1, \cdots, y_{n+1})$ be another $C^2$
coordinate system on $\mathbb{R}^{n+1}$. For any $C^2$ function
$\varphi$, if $\varphi$ satisfies the wave equation with respect
to $x_i$'s if and only if $\varphi$ satisfies the wave equation
with respect to $y_i$'s, then the coordinate transformation is
$C^\infty$.
\end{corollary}

\section{Discussions} \label{section:3}

In elementary wave equation theory, it is known that, if $\varphi$
is a solution of wave equation, then its value at $(\alpha_1,
\cdots, \alpha_{n+1})$ with $\alpha_{n+1}>0$ is completely
determined by the values of $\varphi$ on $C \cap \Sigma$ where
$\Sigma$ is a hyperplane $x_{n+1}=0$ and $C$ is the backward cone
with apex at $(\alpha_1, \cdots, \alpha_{n+1})$. The values of
$\varphi$ on $\Sigma$ outside $C \cap \Sigma$ can not affect the
value of the solution at $(\alpha_1, \cdots, \alpha_{n+1})$. For
this reason, $C \cap \Sigma$ is known as the domain of dependence
of the solution at the point $(\alpha_1, \cdots, \alpha_{n+1})$.

In terms of causality theory, the set $C \cup \Sigma$ is a
causally admissible subset with respect to the Cauchy surface
$\Sigma$ developed in Ref. \cite{CQG5}. To be precise, let
$\Sigma$ be the hyperplane defined by $x_{n+1} = 0$. Then $\Sigma$
is a Cauchy surface and under the causal automorphism $F$,
$F(\Sigma)$ is another Cauchy surface. The condition that ``for
any smooth function $\varphi$, $\sum\limits_{i=1}^n
\frac{\partial^2 \varphi}{\partial x_i^2} - \frac{\partial^2
\varphi}{\partial x_{n+1}^2} = 0$ if and only if
$\sum\limits_{i=1}^n \frac{\partial^2 \varphi}{\partial y_i^2} -
\frac{\partial^2 \varphi}{\partial y_{n+1}^2} = 0$" implies that
there exists a one-to-one correspondence between the domain of
dependence of $\varphi$ at $(\alpha_1, \cdots, \alpha_{n+1})$ with
respect to $\Sigma$ and the domain of dependence of $\varphi$ at
$F(\alpha_1, \cdots, \alpha_{n+1})$ with respect to $F(\Sigma)$.
In other words, the condition $\sum\limits_{i=1}^n
\frac{\partial^2 \varphi}{\partial x_i^2} - \frac{\partial^2
\varphi}{\partial x_{n+1}^2} = 0 \Leftrightarrow
\sum\limits_{i=1}^n \frac{\partial^2 \varphi}{\partial y_i^2} -
\frac{\partial^2 \varphi}{\partial y_{n+1}^2} = 0$ implies the
existence of causally admissible function with respect to $\Sigma$
and $F(\Sigma)$.

As can be seen in theorem \ref{lemma}, \ref{main} and the above
argument, wave equation itself implies the causal relation. In
fact, though we usually define causal relation on Minkowski
spacetime on the basis of physical reason, it can be understood
that, mathematically rigorously, the causal relation can be
defined on the basis of theory of wave equation.

As commented in section 1, causal automorphism on $\mathbb{R}^2_1$
has some peculiar properties in contrast to higher dimensional
Minkowski spacetimes. The method used in this paper does not work
in two-dimensional case since causal automorphism on
two-dimensional spacetime is not necessarily smooth and thus we
can not convert the wave equation written in $x_i$'s into the wave
equation written in $y_i$'s.

It is a well-known fact that, in more than two spacetime
dimensions, any causal automorphism is a conformal
diffeomorphism,(See Ref. \cite{HKM}, \cite{Malament} and
\cite{Fullwood}). Also, it is known that a diffeomorphism $f : (M,
g) \rightarrow (N, h)$ if a conformal transformation if and only
if $g(v,v)=0$ if and only if $h(F_*v, F_*v)=0$. In other words, in
more than two spacetime dimensions, any null vector preserving
diffeomorphism is a causal isomorphism. Null vector preserving
diffeomorphism can be interpreted as that any wave with
propagation speed $c = 1$ must be sent to a wave with propagation
speed $c=1$. In other words, it must be that $\sum\limits_{i=1}^n
\frac{\partial^2 \varphi}{\partial x_i^2} - \frac{\partial^2
\varphi}{\partial x_{n+1}^2} = 0 \Leftrightarrow
\sum\limits_{i=1}^n \frac{\partial^2 \varphi}{\partial y_i^2} -
\frac{\partial^2 \varphi}{\partial y_{n+1}^2} = 0$. If we consider
these, we can see that theorem \ref{main} can be used to give a
new proof of Zeeman's theorem.

\section{Acknowledgement}

This research was supported by Basic Science Research Program
through the National Research Foundation of Korea(NRF) funded by
the Ministry of Education, Science and Technology(2011-0022667).

\end{document}